\documentclass[11pt]{article} 

\usepackage[utf8]{inputenc} 

\usepackage{geometry} 
\usepackage{authblk}

\usepackage[hidelinks]{hyperref}
\usepackage{graphicx} 
\usepackage{amsfonts,amsthm,amsmath,amssymb}
\newcommand{\R}{\mathbb{R}}
\newcommand{\C}{\mathbb{C}}

\newcommand{\bone}{\mathbf{1}}
\newcommand{\N}{\mathbb{N}}
\newcommand{\ket}{\rangle}
\newcommand{\bra}{\langle}

\newcommand{\Li}{\mathcal{L}}

\newcommand{\Ai}{\mathcal{A}}
\newcommand{\Bi}{\mathcal{B}}

\newcommand{\Hi}{\mathcal{H}}

\newcommand{\Mi}{\mathcal{M}}
\newcommand{\Ki}{\mathcal{K}}

\newcommand{\Ei}{\mathcal{E}}
\newcommand{\Si}{\mathcal{S}}
\newcommand{\Fi}{\mathcal{F}}
\newcommand{\Oi}{\mathcal{O}}

\newcommand{\GH}{\mathfrak{H}}

\newcommand{\id}{\operatorname{id}}

\newtheorem{thm}{Theorem}[section]

\newtheorem{Lemma}[thm]{Lemma}
\newtheorem*{Lemma*}{Lemma}

\newtheorem{Corollary}[thm]{Corollary}

\newtheorem{prop}[thm]{Proposition}

\numberwithin{equation}{section}

\theoremstyle{definition}

\newtheorem{Definition}[thm]{Definition}
\newtheorem{dfn}[thm]{Definition}

\newtheorem*{Ack}{Acknowledgment}
\newtheorem{Remark}[thm]{Remark}

\newtheorem*{remark*}{Remark}

\newcommand{\Ad}{\operatorname{Ad}}
\newcommand{\Op}{\operatorname{Op}}
\newcommand{\pd}{\partial}

\usepackage{booktabs} 
\usepackage{array} 
\usepackage{paralist} 
\usepackage{verbatim} 
\usepackage{subfig} 

\usepackage{fancyhdr} 
\evensidemargin=\oddsidemargin
\usepackage{sectsty}
\allsectionsfont{\sffamily\mdseries\upshape} 

\usepackage[nottoc,notlof,notlot]{tocbibind} 
\usepackage[titles,subfigure]{tocloft} 

\title{Electromagnetism in terms of quantum measurements}
\author{Andreas Andersson}
\date{} 
\affil{\small Email: fornjotnr@hotmail.com}
\affil{\footnotesize Wollongong University, School of Mathematics and Applied Statistics, 2522 Wollongong, Australia}
\affil{Mathematics Subject Classification 2010 Primary: 81P15; Secondary: 81R15, 81R60}
\affil{Keywords: Quantum Measurements, POVM, Instrument, Rieffel Deformation, Warped Convolution, Electromagnetism, Gauge Fields, Equilibrium, Liouvillian, Pauli Theorem, Feynman Brackets}

\begin{document}
\maketitle
\abstract

We consider the question whether electromagnetism can be derived from quantum physics of measurements. It turns out that this is possible, both for quantum and classical electromagnetism, if we use more recent innovations such as smearing of observables and simultaneous measurability. In this way we justify the use of von Neumann-type measurement models for physical processes.

We apply operational quantum measurement theory to gain insight in fundamental aspects of quantum physics. Interactions of von Neumann type make the Heisenberg evolution of observables describable using explicit operator deformations. In this way one can obtain quantized electromagnetism as a measurement of a system by another. The relevant deformations (Rieffel deformations) have a mathematically well-defined ``classical" limit which is indeed classical electromagnetism for our choice of interaction. 


\section*{Introduction}
A ``quantum measurement", as described by von Neumann's model and its generalizations, impacts the system sufficiently to make it necessary to regard it as a disturbance. Only the system plus apparatus together forms an approximately closed system. In this paper we will use such models to describe any interaction between quantum system, being it what is in common language meant by a measurement or not. 

For this we use the operational quantum measurement theory \cite{BLM} which during the last two decades has provided a systematic generalization of von Neumann's original formulation of measurements. Just as von Neumann's model, this is one of the conceptually most important components of quantum theory since it gives an operational description of the very interactions themselves. Our task is to see how this is related to the other parts of quantum theory where the interactions are not of tensor product form, as is almost always the case in condensed matter physics (field theory or not). Then there is hope of understanding a great deal of physical phenomena as emerging from these single quantum interactions.    

In fact, the notion of ``sequential measurements", recalled in §\ref{Sequential Measurements}, gives a possibility of comparing the effect of two measurements on each other. The effect of one of the measurements is completely independent on what causes the second measurement. Hence, if experiments are not privileged among interactions, it should be possible to obtain physical forces as ``measurements" done by another (uncontrollable) quantum system. If we place quantum measurement in a central position in quantum physics then we better try to show that it reproduces what quantum theory describes so well, viz. electromagnetic interactions. We will show that this is possible, but we do not think it would be possible without these recent extensions of the von Neumann model. Having this in mind it is maybe understandable that the latter has often been regarded as artificial. This is nevertheless false; for example, the interaction term in the nonrelativistic QED Hamiltonian of an atom interacting with the electromagnetic field takes the same (tensor-product) form as that in the von Neumann model by invoking the dipole approximation only \cite[§4.2]{ALV}. Thus, simple as it might be, this model is really the starting point and with the mentioned addenda it will provide the most direct derivation of electromagnetism. 

However, the main objects in operational measurement theory are not the observables which appear e.g. in quantum field theory. To deal with the same formalism using the conventional observables requires manipulations with unbounded operators. The required technicalities for doing so are taken care of by the identification \cite{An1} of the measurement disturbance as a deformation which can be made mathematically rigorous and relates directly to generators of symmetry transformations and the conventional observables.  

Deformations related to spacetime symmetries and observables have been studied a lot, albeit not as deformations due to quantum interactions. Recently there has been a lot of interest in applying noncommutative geometry to physics, for example due the idea that spacetime itself might be noncommutative, i.e. the idea that the spacetime coordinates $x_0,x_1,x_2,x_3$ do not commute \cite{ADK}. While the main motivation for this (related to quantum gravity) has little relevance for what we will discuss here, it has become much clearer in what way noncommutativities in addition to the standard Heisenberg relation relate to physical effects. Most directly, descriptions of many phenomena in condensed matter physics can be obtained from the case of free particles by, instead of adding interactions by hand to the Hamiltonian or Lagrangian, introducing some nonstandard commutation relations between momenta and/or coordinates (see e.g. \cite{BaGh,DaJe}). Again, this is interesting since the noncommutativity of quantum mechanics has often been linked to a disturbance induced by measurement. What is then the relation between these other noncommutativities, which reproduce forces, and quantum measurements? Since the latter simply describes interactions with small systems of matter, there should exist some relation. In this paper we show that the above noncommutative models can be understood from the theory of quantum measurements initiated by von Neumann.

An interaction between two quantum systems is typically modeled using a unitary operator $W=e^{-i X_\mu\otimes Y^\mu}$ acting on the composite Hilbert space $\Hi\otimes\Ki$ of the two systems. Here $X=(X_1,\dots,X_n)$ and $Y=(Y_1,\dots,Y_n)$ are conventional observables, i.e. unbounded selfadjoint operators on $\Hi$ and $\Ki$ respectively. If $T$ is an observable on $\Hi$ then $W^{-1}(T\otimes\bone)W$ is the corresponding observable after the measurement, before "tracing out" the auxiliary space $\Ki$ using some state. In §\ref{techsection} we recall a result relating this observables to a certain operator deformation introduced quite recently in the context of algebraic quantum field theory \cite{BS,BLS} to incorporate noncommutative effects in physical models. In \cite{An1} it was also shown that the formula
$$
W^{-1}(T\otimes\bone)W=\int_{\R^n}e^{i y^\mu X_\mu}Te^{-iy^\mu X_\mu} \otimes dE^Y(y)
$$
holds even in the case when $T$ is unbounded but satisfies some weak requirements. In particular it includes the case when $A$ is a momentum or coordinate operator, or a polynomial of these (equivalently, a polynomial in annihilation and creation operators).

After recalling this result we devote §\ref{Instruments} and §\ref{Sequential Measurements} to see what the formalism looks like in the language of operational quantum measurement theory. These notions set the stage for §\ref{Minimal Coupling to Gauge Potential} where we discuss the origin of quantum discreteness. We take the viewpoint that the interaction of a quantum system with e.g. a magnetic field can be described by the same quantum measurement theory as a controlled interaction, and we then compare the strength of this interaction to that of an interaction due to an experimental measurement on the system. Since the magnetic field interacts via so much faster (smaller) energy transfers, its presence requires the experimentalist to describe his measurement as a ``joint measurement", a fact which affects the algebra of observables in a significant way. Namely, it forces the minimal substitution $\mathbf{P}\to\mathbf{P}-\mathbf{A}$ of the momentum observables $\mathbf{P}$. Furthermore, this substitution is obtained as a measurement of the coordinate operators $\mathbf{X}$ of the system. Viewing these as symmetry generators one gets the intuitive picture that the coupling to the gauge field $\mathbf{A}$ is equivalent to having the system in constant acceleration, the coordinate operators being up to a constant the generators of boosts in nonrelativistic quantum mechanics. This gives an interpretation of the recent observation made in \cite{Mu1} that the minimal coupling can be rigorously obtained as a deformation.

Working in an equilibrium representation of the observable algebra \cite{HHW} one can obtain a potential energy Hamiltonian in the same fashion as the minimal substitution mentioned above, i.e. from a quantum measurement interaction. This requires avoiding the Pauli theorem saying that there can be no selfadjoint operator corresponding to time \cite{MME}. Fortunately, in the algebraic formulation of equilibrium this no-go result does not apply and we can define a generator of energy translations completely analogous to those of momentum translations.  When we conclude §\ref{Potential Energy} we have obtained the electromagnetic force strength $F=(F_{\mu\nu})$ as the deformation matrix $\Theta$ in the above warped convolution formula; equivalently, electromagnetism is recovered from quantum interactions with $W=e^{-iF_{\mu\nu}X^\mu\otimes X^\nu}$. 

While there exists an abundance of literature on ``quantization", i.e. the concept of either recovering the classical limit $\hbar\to 0$ or to obtain quantum observable algebras starting from the algebra of functions on phase space \cite{Lan} (the former is what makes sense physically), there has been much less effort devoted to relating such a quantization to electromagnetism. Our final aim (§\ref{quantization}) is to discuss the classical limit of quantum-measurement interactions. The warped convolution is closely related to Rieffel deformation, which in turn was motivated by quantization. In particular, there is a well-defined ``classical limit", where the deformation parameter (in our case the coupling strength) goes to zero. In light of the physical meaning given to these constructions in this paper we can make some interesting observations: it follows from Rieffel's deformation quantization that, in the classical limit, the components of the interaction matrix $F_{\mu\nu}$ become the structure constants of a Poisson bracket on the classical algebra of smooth functions of momentum. Then we know that Maxwell's equations of classical electromagnetism emerge quite directly, since such a Poisson bracket has been discussed in relation to Feynman's derivation of electromagnetism \cite{Bra}.

\begin{Ack}
This work was written when the author was affiliated with the Max Planck Institute for Mathematics in the Sciences. The author thanks Gandalf Lechner, Rainer Verch, Adam Rennie and Albert Much for comments about the paper. 
\end{Ack}

\section{Preliminaries}\label{preliminaries}

\subsection{Motivating the von Neumann model}
Suppose that we have two quantum systems, modeled by Hilbert spaces $\Hi$ and $\Ki$ respectively (we speak of $\Hi$ as the ``quantum system" itself to facilitate notation). Suppose that $\Hi$ and $\Ki$ initially do not interact. It is then natural to model the composite system as $\Hi\otimes\Ki$, since then an operator on $\Hi$ cannot influence the properties on $\Ki$, and vice versa, and the spectrum of an operator of the form $X\otimes Y$ is just the Cartesian product of the spectra of $X$ and $Y$. 

Next suppose that there is an interaction between them. If we ignore any other interactions then the total time evolution on $\Hi\otimes\Ki$ should be unitary. The simplest choice is to take a unitary of the form $W=e^{-i X_\mu\otimes Y^\mu}$ since then the (necessarily selfadjoint) operators $X_1,\dots,X_n$ and $Y_1,\dots,X_n$ can easily be given physical interpretations as being associated to physical properties of $\Hi$ and $\Ki$ respectively (which is what we want when modeling an interaction between two quantum systems; otherwise we could, for most purposes, just consider $\Hi\otimes\Ki$ as a unit system from the beginning). 

\begin{dfn}
We refer to an interaction of the form $(\Hi,\Ki,W=e^{-i X_\mu\otimes Y^\mu})$ as a \textbf{von Neumann-model interaction}. 
\end{dfn}
We do \emph{not} make any distinction as to whether the interaction is controlled by a conscient, since that would be a very strong (non-Copernican) assumption and we are trying to minimalize the number of assumptions.

As mentioned in the introduction, an interaction of the form $X\otimes P$, with $X$ and $P$ the coordinate and momentum operator, appears in the Hamiltonian of nonrelativistic quantum electrodynamics in the dipole approximation and, as a result, is used in almost all applications of open quantum systems, to e.g. spectroscopy and other nonequilibrium processes. The interaction is usually taken to be a sum of terms $(a_k+a_k^\dagger)\otimes (c_k+c_k^\dagger)$, and both for Fermionic and Bosonic annilation operators $a_k$ and $c_k$ we have the interpretation of this interaction as being of the form $X\otimes P$. This interpretation is again lost after making the rotating-wave approximation.

\subsection{Quantum measurement theory}\label{Von Neumann Measurements}
The foundations of quantum measurements were laid by von Neumann when he introduced his measurement model \cite{vN}. Generalizations in various directions have been obtained, e.g. to operators with continuous spectra as done by Ozawa \cite{Oza}. During the last two decades the theory of quantum measurements has been developed more systematically in the language of operational quantum theory \cite{BLM}. In this formalism ``observables" are positive operator-valued measures (POVMs).

\begin{Definition}
Let $\Omega$ be a nonempty set and let $\Fi$ be a $\sigma$-algebra of subsets of $\Omega$. A countably additive mapping $E:\Fi\to\Bi(\Hi)$ is called a \textbf{POVM} or \textbf{semispectral measure} if $0\leq E(\Delta)\leq\bone$ for all $\Delta\in\Fi$ (i.e. each $E(\Delta)$ is an \textbf{effect}) and $E(\Omega)=\bone$.
\end{Definition}
\begin{Definition}
A POVM $E:\Fi\to\Bi(\Hi)$ is called a projection-valued measure (\textbf{PVM}) or \textbf{spectral measure} if in addition $E(\Delta)^2=E(\Delta)$ for all $\Delta\in\Fi$ or (equivalently) $E(\Delta)E(\Delta')=0$ whenever $\Delta\cap\Delta'=\emptyset$.
\end{Definition}
Thus if we regard the spectral measure $E^X$ of a selfadjoint operator $X$ as the observable then the POVMs are ``generalized observables". There are very good reasons to argue that POVMs are needed in addition to the PVMs in order to use quantum theory in full power \cite{BLM}, many of which will be very explicit in this paper. Nevertheless, an important effect of the tools we use below (deformations using selfadjoint operators as generators) is that the conventional observables (e.g. multiplication and differentiation operators) can be more directly involved also in this more general formulation of measurements. We shall use the term ``observable" to refer to both selfadjoint operators and POVMs. PVM observables are sometimes also called \textbf{sharp} while POVMs are called \textbf{unsharp}. 

Let $\Mi\subseteq\Bi(\Hi)$ be a von Neumann algebra.
\begin{Definition}\label{meas}
A \textbf{measurement} of an observable $E:\Fi\to\Mi$ is a quintuple $(\Ki,Z,\omega_\Ki,W,f)$ where $\Ki\cong\Hi$ is the separable Hilbert space, $Z$ is a selfadjoint operator on $\Ki$, $\omega_\Ki$ is a normal state on $\Bi(\Ki)$, $W$ is a unitary operator on $\Hi\otimes\Ki$ (the \textbf{time evolution}) and $f:\text{Spec}(Z)\to\Omega$ is a measurable function (into some space $\Omega$) called the \textbf{pointer function}. It is required that
\begin{equation}\label{repeat}
\omega[E(\Delta)]=(\omega\otimes\omega_\Ki)[W^{-1}(\bone\otimes E^Z(f^{-1}(\Delta))W)]
\end{equation}
for all $\omega\in\Mi_*$ and $\Delta\in\Fi$.  

\end{Definition}
The meaning of Definition \ref{meas} is that the elements of $\Mi$ evolve under the measurement according to $A\to W^{-1}(A\otimes\bone)W$, and similarly for those of $\Bi(\Ki)$, and that $\bone\otimes Z$ takes the same values in the final total state as $E$ does in the initial state. Usually one takes $f$ to be the identity function but we shall find in §\ref{Minimal Coupling to Gauge Potential} that we need a nontrivial $f$. With such a measurement scheme $(\Ki,Z,\omega_\Ki,W,f)$ the condition (\ref{repeat}) will be called the \textbf{probability reproducibility condition}. When $E$ is a PVM this condition cannot hold if $E$ has continuous spectrum. Hence the measured observable is either discrete or unsharp (or both) \cite[p.119]{BLM}.

\subsection{Operator deformations}
As a tool for constructing quantum field theories, e.g. for incorporating noncommutative effects of spacetime, Buchholz, Lechner and Summers \cite{BLS}, \cite{BS} introduced a way of deforming an operator $T$ on Hilbert space $\Hi$ to what they called a ``warped convolution" of the operator \cite{BLS}, \cite{BS}. The idea is as follows. For some positive integer $n$, consider an $n$-tuple of commuting selfadjoint operators $P=(P_\mu)=(P_0,P_1,\dots,P_{n-1})$ in $\Hi$; we take as an example the relativistic momentum operator, so $n=4$. It generates an $4$-parameter unitary representation $x\to U(x)$ of spacetime translations in the Hilbert space defined by the physical state. There is thus an action 
$$
\alpha_x(T)=U(x)^{-1}TU(x):=e^{ix\cdot P}Te^{-ix\cdot P}
$$
of $\R^4$ on $\Bi(\Hi)$. For a bounded operator $T$ which is smooth with respect to this action, the \textbf{warped convolution} of $T$ can be defined and equals
\begin{equation}\label{warped}
T_\Theta:=\int_{\R^4}\alpha_{\Theta x}(T)\, dE^P(x),
\end{equation}
where $dE^P(x)$ is the joint spectral measure of the $P_\mu$'s and $\Theta$ is a $4\times 4$ skew-symmetric matrix. In fact, \eqref{warped} makes sense also for certain unbounded operators \cite{Mu1}, a fact that we shall need. We shall only deform selfadjoint operators, and they will always be such that their deformations are again selfadjoint, by \cite{Mu3}. 

Warped convolution turns out to be related to the deformed products developed by Rieffel \cite{BLS},\cite{LW}. This is also very interesting since these product are used in quantization theory (as will be described and used in §\ref{quantization}). 

The formula \eqref{warped} has interesting applications in physics, for example when the generators are the momenta $P_\mu$ but other commuting generators can also be important \cite{Mu2}. Some familiar quantum-mechanical effects where reproduced in \cite{Mu1} by deforming some initial free operators into the desired ones. These results are very fascinating, in particular the fact that, it turns out, deformation with the coordinate operators $X^\mu$ conjugate to the $P_\mu$'s actually reproduces minimal coupling to a gauge field when the matrix $\Theta$ in \eqref{warped} is chosen properly, at least in the nonrelativistic setting (see §\ref{Minimal Coupling to Gauge Potential}). This is intuitive since the generators of Galilean boosts are basically the coordinate operators, and 
$$``\text{boost}\to\text{acceleration}\to\text{force}\to\text{gauge potential}."$$   
The above deformation \eqref{warped} somehow provides a path from symmetries to forces using only the commutation relations of the symmetry group. Can we also understand why this is true? 

Of concern is also the unitarity of the transformation $T\to T_\Theta$. More precisely, it is not true in general that there is a unitary operator $V\in\Bi(\Hi)$ such that an operator $T$ on $\Hi$ can be mapped to $T_\Theta$ by $T\to V^{-1}TV$. But it is known that there are situations when introducing noncommutativity by means of replacing $T$ by $T_\Theta$ can account for the difference between a system without and in the presence of an external force field \cite{Mu1}. The noncommutativity should, as in the case of the Heisenberg relation, come from the interaction between two or more quantum systems. Thus, in addition to the observable algebra of the system, the other player in this interaction (which is not seen in the above description) must be included in order to obtain this unitarity.  

\subsection{Deformations from quantum measurements}\label{techsection}
Let $X=(X_1,\dots X_n)$ and $Y=(Y_1,\dots,Y_n)$ be arbitrary commuting selfadjoint operators in Hilbert spaces $\Hi$ and $\Ki$ respectively. We use summation notation $X_\mu \otimes Y^\mu:=\sum_{k=1}^nX_k\otimes Y_k$ etc. We refer to \cite{An1} for the technical details of the following. 
\begin{thm}\label{maintheorem}
Let $T$ be an operator acting in $\Hi$ satisfying certain conditions (for example $T$ can be any polynomial in the coordinate or momentum operators). Then on a certain dense subspace of $\Hi\otimes\Ki$ we have (in the weak sense) the equalities 
\begin{align*}
e^{iX_\mu\otimes Y^\mu}(T\otimes\bone)e^{-iX_\mu\otimes Y^\mu}
=\int_{\R^n}e^{i y^\mu X_\mu}Te^{-iy^\mu X_\mu} \otimes dE^Y(y),
\end{align*}
where $\R^n\ni\to y\to dE^Y(y)$ is the joint spectral measure of $Y$. 
\end{thm}
Note that if $\Ki$ is just another copy of $\Hi$ and $Y$ just another copy of $X$ then the right-hand side of the above formula looks exactly like the warped convolution \eqref{warped} except for the tensor product factor. 

Thus, if we consider the case when $W:=e^{-iX_\mu\otimes X^\mu}$ plays the role of a measurement time evolution, the post-measurement observable $W^*(T\otimes\bone)W$ in $\Hi\otimes\Ki$ obtained from $T\otimes\bone$ is like a warped convolution using the $X^\mu$'s as generators. There is a nontrivial distinction because of the tensor product that will be discussed in detail in §\ref{Minimal Coupling to Gauge Potential}. The important point is that when $\Ki$ is ignored, the warped convolution is a good way of modeling the deformation due to an interaction. This becomes interesting when one considers some $C^*$-algebra of operators $A$ because one then has the relation to Rieffel deformation. This will be essential for obtaining classical electromagnetism from quantum measurements. 

If $W$ is viewed as a measurement then the commuting operators $X_1,\dots X_n$ are the ones intended to be measured on the quantum system. These are usually not the same as the "measured observables", as we recall next. 

\subsection{Instruments}\label{Instruments}
Let $T$ be an operator on $\Hi$. So far we have discussed the element $W^{-1}(T\otimes\bone)W$ corresponding to $T$ after an interaction $W=\text{exp}(-iX\otimes Y)$ with some other system $\Bi(\Ki)$. But $T\otimes\bone$ is an operator on the composite system $\Hi\otimes\Ki$. The evolution of $T$ is obtained after choosing an initial state $\omega_\Ki$ on $\Bi(\Ki)$ and evaluating $W^{-1}(T\otimes\bone)W$ in $\bone\otimes\omega_\Ki$. This last step is similar to the partial trace operation on states.
 
If $\omega_\Ki\in\Bi(\Ki)_*$ is the initial state on $\Ki$ then the time evolution of an element $T\in\Bi(\Hi)$ is given by 
$$
T\to\int_{\R}e^{iyX}Te^{-iyX} \omega_\Ki[dE^Y(y)].
$$
Now it may be that the outcome of the measurement is recorded by measuring the pointer observable $E^Z$ ``conjugate" to $E^Y$, i.e. $[Y,Z]=i\bone$. In that case the evolution of $A$ can be made more precise; it is zoomed in using the outcome of the measurement. For this we use the notion of ``instruments" \cite{DL}. Namely, for all Borel subsets $\Delta$ of $\R$ we have the map
$$
\Ei^*_{\Delta}:\Bi(\Hi)\to\Bi(\Hi), \qquad T\to (\id\otimes\omega_\Ki)[W^{-1}(T\otimes E^Z(f^{-1}(\Delta)))W],
$$
which defines the \textbf{dual} $\Ei^*:\Delta\to\Ei^*_\Delta$ of the instrument of the interaction described by $X\otimes Y$ (here $f:\text{Spec}(Z)\to\text{Spec}(X)$ is a function relating the spectra as in Definition \ref{meas}). 
\begin{Remark}
In terms of the completely positive map $\Ei_\Delta^*$, the probability reproducibility condition \eqref{repeat} takes the form
$$
\omega[\Ei_\Delta^*(\bone)]=\omega[E(\Delta)], \qquad \forall \omega\in\Mi_*,\Delta\in\Fi.
$$
\end{Remark}
The \textbf{instrument} $\Delta\to(\Ei_\Delta:\Bi(\Hi)_*\to\Bi(\Hi)_*)$ is then defined via
$$
(\Ei_{\Delta}(\rho))(T)=\rho(\Ei^*_{\Delta}(T)), \qquad T\in\Bi(\Hi).
$$
Thus, the map $\Ei_{\Delta}$ on states corresponds to the Schrödinger picture while the map $\Ei_{\Delta}^* $ on observables corresponds to the Heisenberg picture.
For our purposes however, the most important use of instrument is that it can show us how the statistics of the measured observable are generally not the same as one might have guessed. 
\begin{Definition}\label{measobsdef}
Let $\Ei:\Bi(\Hi)_*\to\Bi(\Hi)_*$ be an instrument. The \textbf{measured observable} $E^{\Ei}$  associated to $\Ei$ is defined by
$$
E^{\Ei}(\Delta):=\Ei_\Delta^*(\bone).
$$
\end{Definition}
\begin{Remark}
For a given instrument $\Ei$ the measured observable $E^\Ei$ is unique. On the other hand, there are many instruments which define the same observable $E^\Ei$.  
\end{Remark}
Since we have obtained an explicit formula for the deformed observables, the observable associated to $\Ei$ can be calculated explicitly \cite{An1}. 
\begin{Corollary}
Let $(\Ki,Z,\omega_\Ki,e^{-i\kappa X\otimes Y},f)$ be a measurement of a sharp observable $X$ as in Definition \ref{meas} and assume that $[Y,Z]=i\bone$. Then the measured observable is given by  
\begin{equation}\label{measobs}
E^{\Ei}_\kappa(\Delta)=\int_{\R}\omega_\Ki[dE^Z(f^{-1}(\Delta-\kappa x))]\, dE^X(x). 
\end{equation}
\end{Corollary}
Therefore, unless $\omega_\Ki[dE^Z(f^{-1}(\Delta-\kappa x))]=\chi_\Delta$ for all $x\in\text{Spec}(X)$, the measured observable will not be equal to the one ``intended" to be measured, i.e. the spectral measure $E^X$ of the operator $X$ appearing in $e^{-i\kappa X\otimes Y}$. When we have $[Y,Z]=i\bone$, this can only happen if $\kappa=0$. On the other hand, if $[Y,Z]=0$ then $\Ei_\Delta^*(\bone)=\omega_\Ki[E^Z(f^{-1}(\Delta))]\bone$ for all $\Delta$ so that only a multiple of the identity can be measured (which is usually far from $E^X$!). This manifests the trade-off between accuracy and disturbance. We will see an example in §\ref{Minimal Coupling to Gauge Potential} where it is important which observable is actually measured. 

\subsection{Sequential measurements}\label{Sequential Measurements}
The above formalism can be used to describe a subsequent measurement of a second observable, something which is nontrivial because of the disturbance caused by the first interaction. For example, after having measured position, it would be interesting to see how this affects a measurement of momentum, since by the commutation relations the momentum is deformed by the position measurement. Again the notion of instrument is well adapted for this task. 

Let $E_1$ and $E_2$  be two POVMs on $\Hi$ associated to instruments $\Ei^1$ and $\Ei^2$, respectively, so that for every trace-class operator $\rho$ on $\Bi(\Hi)$ and Borel set $\Delta\subset\R$ we have
$$
\text{Tr}\{\rho E_1(\Delta)\}=\text{Tr}\{\Ei_\Delta^1(\rho)\},\quad\quad \text{Tr}\{\rho E_2(\Delta)\}=\text{Tr}\{\Ei_\Delta^2(\rho)\}.
$$ 
If $E_1$ is measured with outcome in $\Delta_1$ followed by $E_2$ in $\Delta_2$ then the instrument $\Ei^2\circ\Ei^1$ of the composite measurement takes the value $\Ei_{\Delta_2}^2\circ \Ei_{\Delta_1}^1$. There is a unique observable $F$ associated to this element, given by \cite{CHT2}
$$
F(\Delta_1\times\Delta_2):=(\Ei_{\Delta_1}^1)^*[E_2(\Delta_2)].
$$
Of relevance are the so-called marginals given by
$$
F_1(\Delta_1):=F(\Delta_1\times\R)=(\Ei_{\Delta_1}^1)^*[E_2(\R)]=E_1(\Delta_1),
$$
$$
F_2(\Delta_2):=F(\R\times\Delta_2)=(\Ei_{\R}^1)^*[E_2(\Delta_2)]\equiv E_2'(\Delta_2),
$$
so that the second marginal is not $E_2$ but a deformed version $E_2'$. The case when $E_2=E_2'$ means that each effect $E_2(\Delta_2)$ of $E_2$ is a fixed point of $(\Ei_{\R}^1)^*$. In \cite{HW} this scenario is summarized by saying that $E_1$ can be measured ``without disturbing" $E_2$. For the kind of observables we consider here, i.e. the spectral measures of selfadjoint operators, $E_1$ does not disturb $E_2$ if and only if $[E_1(\Delta),E_2(\Delta)]=0$ for all $\Delta\subset\R$. This is just another way of seeing that if $\Mi$ is the observable algebra of interest which we deform by generators of a maximal abelian subalgebra $\Ai\subset\Mi$, then the deformation of $B\in\Mi$ is nontrivial precisely when $B\notin\Ai$. This is easy to see in terms of projections, since $\Ai$ is generated by its projections. 

\begin{Definition} Let $F_1$ and $F_2$ be POVMs  on $(\Omega_1,\Fi_1)$ and $(\Omega_2,\Fi_2)$  respectively with values in $\Bi(\Hi)$. We say that $F_1$ and $F_2$ are \textbf{simultaneously measurable} if there exists a POVM $F$ on $(\Omega_1\times\Omega_2,\Fi_1\times \Fi_2)$ with marginals $F_1,F_2$. 
\end{Definition}

\section{Minimal coupling to gauge field}\label{Minimal Coupling to Gauge Potential}
Starting with the free Hamiltonian of a charged mass-$m$ particle in three dimensions,
\begin{equation*}
H_0:=\frac{\mathbf{p}^2}{2m},
\end{equation*}
the presence of a magnetic field requires changing the momentum $\mathbf{p}$ to $\mathbf{p}+q\mathbf{A}$,
\begin{equation}\label{Hqclassical}
H_q:=\frac{(\mathbf{p}+q\mathbf{A})^2}{2m},
\end{equation}
where $q$ is the charge and $\mathbf{A}$ is the electromagnetic vector potential \cite{Ja}. In quantum field theory language,
the field creating the given particle species has been ``minimally coupled" to the gauge field $\mathbf{A}$. 

From now on, $H_q$ will denote the canonically quantized version of \eqref{Hqclassical}, so $\mathbf{p}$ is replaced by a triple $\mathbf{P}=(P_1,P_2,P_3)$ of operators on a Hilbert space $\Hi$ in which also acts a triple $\mathbf{X}=(X_1,X_2,X_3)$ of operators such that $[P_j,X_k]=i$. 

It turns out that $H_q$ can be obtained from $H_0$ by a certain (nonrelativistic) ``boosting" of the system. Namely, consider a skew-symmetric $3\times 3$ matrix $\Theta$ of the form 
\begin{equation}\label{Thetaisthis}
\Theta^{jk}=\epsilon^{ijk}B_i
\end{equation}
for some vector $\textbf{B}=(B_1,B_2,B_3)$.  
\begin{Lemma}[{\cite[Lemma 4.1]{Mu1}}]
Define the real skewsymmetric $3\times 3$ matrix $\Theta$ by Equation \eqref{Thetaisthis} and define an action of $\R^3$ on operators $A$ in $\Hi$ by 
\begin{equation}\label{Muchaction}
\alpha_p(A):=e^{-i\mathbf{p}\cdot\mathbf{X}}Ae^{i\mathbf{p}\cdot\mathbf{X}},\qquad \forall \mathbf{p}\in\R^3.
\end{equation}
Then the momentum operator deforms as in \eqref{warped} into 
\begin{equation}\label{Muchobs}
\mathbf{P}_{q\Theta}=\mathbf{P}+q\mathbf{A},
\end{equation}
where $\mathbf{A}:=\epsilon^{j,k,l}B_kX_l$. 
\end{Lemma}
Recall that the coordinate operators are up to a constant the generators of boosts in nonrelativistic mechanics. 
Thus \eqref{Muchobs} says that applying a boost to the system, via the action \eqref{Muchaction}, results in the minimal substitution \eqref{Hqclassical} on quantum level. 

Here we shall try to elucidate the meaning of this peculiarity. The idea is that, in view of Theorem \ref{maintheorem}, it seems plausible that we can accomplish such a deformation as a (Heisenberg-picture) unitary evolution by enlarging the Hilbert space. 

So let $\Hi$ be the Hilbert space of the free (spinless) particle, so that the most natural way to view $\Hi$ is as $L^2(\R^3)$. The state space of the magnetic field is taken to be $\Ki=L^2(\R^3)$ as well. We start with the von Neumann interaction
\begin{equation}\label{Twist}
W_q=e^{-iq\epsilon^{ijk}B_iX_j\otimes X_k},
\end{equation}
where the $X_j$'s are the components of the coordinate operator $\mathbf{X}=(X_1,X_2,X_3)$ and Einstein summation is implicit for $j,k=1,2,3$. This interaction looks a little bit unfamiliar in the context of measurements since we take $\mathbf{X}$ on both factors. We could equivalently have taken $\mathbf{X}\otimes\textbf{P}$ but the particle-field interaction is more symmetric in the above notation. The coupling matrix $\Theta$ will be identified with the magnetic field and the magnetic two-form acts on tangent vectors, hence the position operators on both factors. 

In terms of the choice of ``vector potential" $\mathbf{A}$ given by $A^k:=\Theta^{jk}X_j$ (recall $\mathbf{B}=\boldsymbol\nabla\times\mathbf{A}$), the interaction unitary \eqref{Twist} is
$$
W_q=e^{-iq\mathbf{X}\otimes\mathbf{A}}:=e^{-iqX_1\otimes A_1}e^{-iqX_2\otimes A_2}e^{-iqX_3\otimes A_3}.
$$

To motivate the choice \eqref{Twist} we assume that electromagnetic interactions results from mutual exchange in energy between the magnetic field and the particle. The system described by $\Ki$ consists of particles in states very similar to that of the single particle in the plane described by $\Hi$ (or else there could be no interaction); they are treated on an equal footing as is seen explicitly from the symmetry in the  interaction operator \eqref{Twist}. 

\begin{Lemma}\label{potthm}
With $W_q$ given by \eqref{Twist} and $\Theta^{jk}X_j=A^k$ for $\mathbf{A}=(A_1,A_2,A_3)$,
$$
W_q^{-1}(\mathbf{P}\otimes\bone)W_q=\mathbf{P}\otimes\bone+\bone\otimes q\mathbf{A}
$$
\end{Lemma}
\begin{proof} Straightforward in view of Theorem \ref{maintheorem}. 
\end{proof}
Thus, for the quantum system $\Hi$ of a free particle in $\R^3$, a measurement of the coordinate operator $\mathbf{X}$ on $\Hi$ using $W_q$ gives the minimal substitution.

Lemma \ref{potthm} is an open-system analogue to the equality \eqref{Muchobs} obtained in \cite{Mu1} without including the part $\Ki$ in the description. The operator $\mathbf{A}$ does not act on the same Hilbert space as $\mathbf{P}$ in the tensor product picture, in accordance with how the field-particle-interaction are usually treated in spin-boson-type models etc. On the other hand, in some condensed matter models it is crucial that these operators actually do fail to commute. Namely, the operator $(\mathbf{P}+q\mathbf{A})^2/2m$ has discrete spectrum as a consequence of $\mathbf{A}$ acting on the same Hilbert space and not commuting with $\mathbf{P}$. This is well-known to be the physically correct result in some cases: an electron moving in a plane with a perpendicular magnetic field can only adopt a discrete set energy states. However, the free Hamiltonian $H_0=\mathbf{P}^2/2m$ has continuous spectrum and since a unitary transformation preserves the spectrum it is clear that the same is true for $W_q^{-1}(H_0\otimes\bone)W_q$ given in Lemma \ref{potthm}. 

So, while the same term $q\mathbf{A}$ appears quite satisfactorily by viewing the field as a second quantum system, it seems as if the tensor product structure prevents a proper mathematical framework for describing all physical effects. The way in which the vector potential is dealt with as an operator is not universal as to if it acts on the same Hilbert space as the particle operators or not. When there is a third interacting component (e.g. experimentalist) the gauge potential and the particle operators are assumed to act on the same Hilbert space and one achieves discretization in energy. When the experimentalist controls the field interaction (e.g. as in quantum optics) there should be no noncommutativity and a tensor product is used\footnote{Or, e.g. as for the Hamiltonian of nonrelativistic quantum electrodynamics, the vector potential acts on both spaces but $[\mathbf{P},\mathbf{A}]=0$ because of the representation of $\mathbf{P}$ and the Coulomb gauge $\boldsymbol\nabla\cdot\mathbf{A}=0$.}. 
 
Using a general open-system approach only there seems to be some arbitrariness in choice between commutativity and noncommutativity of $\mathbf{P}$ and $\mathbf{A}$ when more that two systems interact. The tensor product remains unless removed by hand or disregarded from the beginning. Still, we shall see that the expected discreteness appears from quantum measurements anyway, i.e. that the qualitative picture is always correct. 

\subsection{Discreteness}\label{discsec}
Now that we have obtained such a nice picture of the field interaction in terms of energy transfer, one may ask if it is possible to get an intuitive picture also for the cause of the discretization of energy levels that occur in some condensed matter systems of particles in constant magnetic fields. The presence of such energy transfers (magnetic field) somehow affects the way in which another system interacting with the particle can abstract or donate energy to the particle. We shall discuss this using some aspects of the sequential measurements mentioned in §\ref{Sequential Measurements}. 

For ease of notation, define the map
$$
\tilde{B}:\R^3\to\R^3,\qquad \tilde{B}\textbf{y}:= \mathbf{B}\wedge\mathbf{y},
$$
where $\mathbf{B}\wedge\textbf{y}$ is viewed as a function of tangent vectors $\mathbf{x}\in\R^3$ so that the scalar product $\tilde{B}\mathbf{y}\cdot\mathbf{x}$ is the dual pairing between cotangent and tangent vectors. In this notion the above coupling is $W=e^{-i\tilde{B}\mathbf{X}\otimes\mathbf{X}}$. For the measurement with $W$ we note that the pointer function from Definition \ref{meas} is given by $f=\tilde{B}^{-1}$.

Let us again view the apparatus Hilbert space as $\Ki=L^2(\R^3)$ and the total Hilbert space as $\Hi\otimes\Ki=L^2(\R^3)\otimes L^2(\R^3)$. We spectrally decompose $\textbf{X}\otimes\bone$ and $\bone\otimes\mathbf{X}$ as
$$
\mathbf{X}\otimes\bone=\int_{\R^3}\mathbf{x}\, dE^{\mathbf{X}\otimes\bone}(\mathbf{x}),\quad\quad  \bone\otimes\mathbf{X}=\int_{\R^3}\mathbf{x}\, dE^{\bone\otimes\textbf{X}}(\mathbf{x}).
$$
In what follows we apply constructions summarized in \cite{BL}. Suppose that the initial state of the magnetic field is a vector state $\omega_\psi$ (the below formulae can easily be adjusted to general initial states). We view $\psi$ as a function on momentum space $(\R^3)^*=\R^3$ and we assume for later purposes that $\psi$ is continuous and has compact support. The dual $\Ei_{\Delta}^*$ of the measurement instrument, as discussed in the preliminaries, is defined in terms of the state $\id\otimes \omega_\psi$: If we write the interaction as $W=e^{-i\tilde{B}\textbf{X}\otimes\textbf{X}}=\int_{\R^3} dE^{\textbf{X}\otimes\bone}(\textbf{x})\otimes e^{-i\tilde{B}\textbf{x}\textbf{X}}  $ (using the above map $\tilde{B}$) then for any observable $\textbf{G}$ on the particle space $\Hi$, Equation \eqref{measobs} gives
\begin{align*}
\Ei_{\Delta}^*(\textbf{G})&=(\bone\otimes\omega_\psi)(W^{-1}(\textbf{G}\otimes E^{\bone\otimes\textbf{P}}(\tilde{B}\Delta))W)
\\&=\int_{\R^3} \int_{\R^3}\, dE^{\textbf{X}\otimes\bone}(\textbf{y})\textbf{G}\, dE^{\textbf{X}\otimes\bone}(\textbf{z})\bra\psi|e^{i\tilde{B}\textbf{y}\cdot(\bone\otimes\textbf{X})}E^{\bone\otimes\textbf{P}}(\tilde{B}\Delta)e^{-i\tilde{B}\textbf{z}\cdot(\bone\otimes\textbf{X})}\psi\ket
\\&=\int_{\R^3} \int_{\R^3} \psi^*(\tilde{B}(\Delta-\textbf{y}))\, dE^{\textbf{X}\otimes\bone}(\textbf{y})\textbf{G}\, dE^{\textbf{X}\otimes\bone}(\textbf{z})\psi(\tilde{B}(\Delta-\textbf{z}))
\\&=\int_{\Delta}K_{\textbf{x}}^*\textbf{G}K_{\textbf{x}},
\end{align*}
where we have defined the operators 
$$
K_{\textbf{x}}=\bra \textbf{x}|W\psi\ket=\psi(\tilde{B}(\textbf{x}-\textbf{X})):=\int_{\R^3}\psi(\tilde{B}(\textbf{x}-\textbf{y}))\, dE^{\textbf{X}\otimes\bone}(\textbf{y})
$$
The ``measured" observable (determined by the measurement process which gave the minimal coupling above), obtained from  Definition \ref{measobsdef} as
$$
\Delta\to\Ei_\Delta^*(\bone)=|\psi(\tilde{B}(\Delta-\textbf{X}))|^2
=\int_{\Delta}K_{\textbf{x}}^*K_{\textbf{x}},
$$ 
 is then not the spectral measure of $\textbf{X}$ but an unsharp version: the PVM $E^{\textbf{X}\otimes\bone}$ has been replaced by the POVM $\mu^B*E^{\textbf{X}\otimes\bone}$ with effects
$$
\big(\mu^B*E^{\textbf{X}\otimes\bone})(\Delta\big)=\int_{\R^3}\mu^B(\Delta-\textbf{y})\, dE^{\textbf{X}\otimes\bone}(\textbf{y})
$$
smeared by convolution with the probability measure $\mu^B(\Delta):=\bra\psi|E^{\bone\otimes\textbf{P}}(\tilde{B}\Delta)|\psi\ket$. For completeness, note that the following explicit formulae hold: 
\begin{align*}
\big(\mu^B*E^{\textbf{X}\otimes\bone}\big)(\Delta)&=\int_{\R^3}\mu^B(\Delta-\textbf{y}) dE^{\textbf{X}\otimes\bone}(\textbf{y})
\\&=\int_{\R^3}\bra\psi|E^{\bone\otimes\textbf{P}}(\tilde{B}(\Delta-\textbf{y}))\psi\ket \, dE^{\textbf{X}\otimes\bone}(\textbf{y})
\\&=\int_{\R^3}\bra\psi|e^{i\tilde{B}\textbf{y}\cdot(\bone\otimes\textbf{X})}E^{\bone\otimes\textbf{P}}(\tilde{B}\Delta)e^{-i\tilde{B}\textbf{y}\cdot(\bone\otimes\textbf{X})}\psi\ket\, dE^{\textbf{X}\otimes\bone}(\textbf{y})
\\&=\Ei_\Delta^*(\bone)
\end{align*}
[while in fact $(\textbf{X}\otimes\bone)*\mu^B:=\int\textbf{x}(\mu^B*dE^{\textbf{X}\otimes\bone})=\textbf{X}\otimes\bone$, leaving the abelian algebra generated by the position operators unchanged]. Note that 
$$
K_{\textbf{x}}^*K_{\textbf{x}}=|\psi(\tilde{B}(\textbf{x}-\textbf{X}))|^2
$$
are effect operators related via $\tilde{B}$ to the spatial distribution of matter (the operators $K_{\textbf{x}}$ are of Kraus type in the sense of \cite{Hol}). 

We obtain the following interpretation. The operator $\mu^B*E^{\mathbf{X}\otimes\bone}(\Delta)$ is obtained from $E^{\mathbf{X}\otimes\bone}(\Delta)$ by averaging over the operators $dE^{\mathbf{X}}(\mathbf{y})$ with weights $\bra\psi|dE^{\bone\otimes\mathbf{P}}(\tilde{B}(\Delta-\mathbf{y}))\psi\ket$ determined by the state of the magnetic field source, as well as the coupling $\tilde{B}$. 
We assume that $\psi$ vanishes at infinity, in which case 
$$
\lim_{B\to\infty}\bra\psi|e^{i\tilde{B}\mathbf{y}\cdot\textbf{X}}dE^{\mathbf{P}}(\tilde{B}\mathbf{x})e^{-i\tilde{B}\mathbf{y}\cdot\textbf{X}}\psi\ket=\delta(\mathbf{y}-\mathbf{x}), \quad\quad \forall\mathbf{x},\mathbf{y}\in\R^3,
$$
saying that in the unrealistic situation with infinite strength of the particle-field interaction, the position of the particle is not smeared. As we discuss next, this would imply that it is impossible to determine the energy of the electron. This limiting case serves merely to give intuition for the interesting cases with finite $\|\mathbf{B}\|$. The stronger the $\mathbf{B}$-field, the finer the discretization of the position but the coarser the energy, as we shall now see.

The observables $H_0$ and $\textbf{X}$ are sharp and they do not commute, so it would appear as if the interaction of the particle with the magnetic field makes it impossible to measure the energy of the particle \cite{HRS}. But there is no contradiction because it is well known that smearing the observable $H_0$ makes it jointly measurable with $\textbf{X}$. This smearing is precisely what causes the discretization of the free Hamiltonian in the Landau problem, as we now explain.

We have already seen that the observable actually ``measured" (by the field) is not $\textbf{X}$ but a smeared version. Therefore, all that needs to be done to obtain an Hamiltonian which is measurable in the presence of the magnetic field is to smear $H$ as well. This conclusion comes from the following important result, where $\Delta \to Q(\Delta)$ and $\Delta\to P(\Delta)$ denote the usual sharp position and momentum observables taking values in $\Bi(\Hi)$ where $\Hi=L^2(\R)$.

\begin{Lemma} [Recalled from \cite{HRS},\cite{CHT}]\label{Finishlemma}
Let $\chi$ and $\eta$ be probability measures on $\R$ and define position and momentum observables $Q_\chi$ and $P_\eta$ by
$$
Q_\chi(\Delta):=\int_{\R}Q(\Delta-q)\,d \chi(q), \quad\quad P_\eta(\Delta):=\int_{\R}P(\Delta-p)\, d\eta(p),
$$
respectively. These are simultaneously measurable if and only if there exists a positive trace one operator $T\in\Bi(\Hi)$ such that
$$
\chi(\Delta)=\mathrm{Tr}\{Q(\Delta)T\},\quad\quad \eta(\Delta)=\mathrm{Tr}\{P(\Delta)T\}
$$
for all Borel subsets $\Delta\subseteq\R$.
\end{Lemma} 
\begin{thm}\label{potthm2}
With the choice of measurement scheme $(\Ki,\bone\otimes\mathbf{P},\omega_\psi,W_q,\tilde{B}^{-1})$ of the coordinate operator $\mathbf{X}$ on $\Hi$ as outlined above, a energy observable on $\Hi$ (i.e. a POVM on $\Hi$ whose first moment is the selfadjoint operator $H_0$) is simultaneously measurable if and only if its effects are given by
\begin{align*}
\tilde{H}(\Delta)&:=\int_{\R^3} |\tilde{\psi}(\tilde{B}^{-1}\mathbf{p})|^2\, dE^{H_0}\Big(\Delta-\frac{\mathbf{p}^2}{2m}\Big)
\end{align*}
for all Borel $\Delta\subset\R$, where $\tilde{\psi}(\mathbf{x})$ is the Fourier transform of $\psi(\mathbf{p})$. 
\end{thm}
\begin{proof}
From the expression $d\mu^B(\mathbf{x})=\bra\psi|dE^{\bone\otimes\mathbf{P}}(\tilde{B}\mathbf{x})|\psi\ket=|\psi(\tilde{B}\textbf{x})|^2$ of the convolution measure it is clear that the trace operator $T$ from the above lemma is $T=|\psi_{\tilde{B}}\ket\bra\psi_{\tilde{B}}|$ in our case, where $\psi_{\tilde{B}}(\textbf{x}):=\psi(\tilde{B}\textbf{x})$. Now the formulae become clear when we explicitly write out the Fourier transform,
$$
\psi(\tilde{B}\mathbf{x})=\int_{\R^3}\tilde{\psi}(\mathbf{y})e^{i\mathbf{y}\cdot\tilde{B}\mathbf{x}}\, d\mathbf{y}
=\int_{\R^3}\tilde{\psi}(\tilde{B}^{-1}\mathbf{p})e^{i\mathbf{p}\cdot\mathbf{x}}\, d\mathbf{p}
$$
since, using Lemma \ref{Finishlemma}, we then see that the momentum observables $\tilde{\mathbf{P}}$ on $\Hi$ are simultaneously measurable if 
\begin{align*}
\tilde{\mathbf{P}}(\Delta)&=\int_{\R^3}\bra\psi|dE^{\bone\otimes\mathbf{X}}(\tilde{B}^{-1}\mathbf{p})|\psi\ket\, dE^{\mathbf{P}\otimes\bone}(\Delta-\mathbf{p})
\\&=\int_{\R^3} |\tilde{\psi}(\tilde{B}^{-1}\mathbf{p})|^2\, dE^{\mathbf{P}\otimes\bone}(\Delta-\mathbf{p})
\end{align*}
for each Borel set $\Delta\subset\R^3$. 
\end{proof}
The mapping
$$
\mu^B:\R^3\times\Fi(\R^3)\to\Bi(\Hi),\qquad (\mathbf{x},\Delta)\to d\mu^B(\Delta-\mathbf{x})
$$
is continuous and each $\Delta\to d\mu^B(\Delta-\mathbf{x})$ is a probability measure (so $\mu^B$ is a ``confidence measure"). 
Viewed in another way, for each Borel set $\Delta$, the map $\mathbf{x}\to d\mu^B(\Delta-\mathbf{x})$ is a ``fuzzy event". Following \cite{HLY} we denote this map by $\tilde{\Delta}$ and similarly in momentum space. Then we can write the relation between $\tilde{\mathbf{P}}$ and $E^{\mathbf{P}\otimes\bone}$ as
$$
\tilde{\mathbf{P}}(\Delta)=E^{\mathbf{P}\otimes\bone}(\tilde{\Delta}).
$$
Now comes the problem of recording the energy measurement. That is, we would like to see what outcomes we could have, and so we should try to replace the fuzzy set $\tilde{\Delta}$ by some ordinary set such that we still get the probabilities described by $E^{\mathbf{P}\otimes\bone}$. This is the problem of ``reading the scale" \cite[III.2.4]{BLM}. We cover momentum space $\R^3$ by disjoint cubes $\Delta_n$ labeled by $n\in\N$ and we pick one point $\mathbf{p}^{(n)}$ in the center of each $\Delta_n$. The size of $\Delta_n$ is chosen such that the support of the shifted function $\textbf{x}\to\tilde{\psi}(\textbf{x}-\tilde{B}^{-1}\textbf{p}^{(n)})$ lies in $\Delta_n$. Thus the side length of $\Delta_n$ can be taken to be the diameter of the support $\Delta(\tilde{\psi})$ of $\tilde{\psi}$ multiplied by the field strength $B$, independent of $n$. Define the discrete observable
$$
\mathbf{P}^{\{\Delta_n\}}(n):=E^{\mathbf{P}\otimes\bone}(\Delta_n),\qquad \forall n\in\N.
$$
Compose the pointer function $\tilde{B}^{-1}$ with a map $g_B:\R^3\to\R^3$ such that
$$
g_B(\Delta_n)=\mathbf{p}^{(n)},\qquad g(\R^3\setminus \Delta_n)\cap \{\mathbf{p}^{(1)},\mathbf{p}^{(2)},\dots\}=\emptyset.
$$
Going over to the energy operator $H_0$, the cubes are replaced by disjoint intervals $I_n\subset\R$ with length proportional to the field strength $B$. Write $H^{\{I_n\}}(n):=E^{H_0}(I_n)$ for this rescaled Hamiltonian. Then one shows the following (c.f. \cite[III.2.6]{BLM}).

\begin{prop}\label{rescaling}
$H^{\{I_n\}}$ is the measured observable in the energy measurement from Theorem \ref{potthm2} with the modified pointer function $g_B\circ\tilde{B}$. 
\end{prop}
Thus, as long as we only obtain the energies $n\omega$ for $n\in\N_0$, where $\omega:=B/2m$, the energy can be measured in the presence of the $\mathbf{X}$-measurement that we postulated.

Theorem \ref{potthm2} gives a condition which the momenta (hence energy) of a free electron in the presence of the magnetic field must satisfy. The stronger the coupling to the magnetic field the less is the position operator affected by this coupling; the "position measurement" by the field is accurate. On the other hand, the stronger the $\mathbf{B}$-interaction the more smeared will be the subsequent (or rather joint) energy-measurement. Regarding the measurability we also found above that $(\mathbf{P}-q\mathbf{A})^2$ suffices. We have then realized that the gauge field is equivalent to a smearing leading to the discretization as understood in terms of quantum measurements. The interaction with the magnetic field, \emph{which  is always present in observations on slower time-scales},  leaves the system with an algebra of discrete operators. 

There is an important distinction between this smeared observable and the smearing that one usually has in mind when such a ``fuzzification" of a sharp observable is considered, i.e. the usual picture of experimental error and noise. In the present case the magnetic field interaction is supposed to occur very fast and with extremely many repetitions. Thus it is too regular for being regarded as noise; it will give the same influence on all our measurements.

Finally let us stress that the same argument can be applied whenever there is a comparison of two interactions, accounting in this way for quantum discreteness. 

\section{Potential energy}\label{Potential Energy}
\subsection{Pauli no-go theorem circumvented}\label{Paulisec}
In the usual formulations of quantum and classical physics the Hamiltonian $H$ is bounded from below. Hence the existence of a selfadjoint operator $X_0$ on the same Hilbert space with $[X_0,H]=i\bone$ would be a contradiction \cite{Pau}, \cite{MME},\cite[Sec. 13.2]{Mor}, as seen from the covariance property $e^{i\lambda X_0}E^H(\Delta)e^{-i\lambda X_0}=E^H(\Delta+\lambda)$ of the spectral measure of $H$. This observation is called the \textbf{Pauli theorem}. 

However, if the system is in an equilibrium state $\omega$ then the generator of time translations in the natural choice of Hilbert space (the so-called GNS representation $\Hi_\omega$ of the observable algebra $\Mi$ associated to $\omega$) typically has the whole line $\R$ as spectrum. There are several unitary groups acting on $\Hi_\omega$ implementing the time translations on the image of $\Mi$ as an algebra of operators on $\Hi_\omega$. The naive choice of such a group, for avoiding the Pauli theorem, would be to look for an operator with purely absolutely continuous spectrum $\R$, since this would ensure the existence of a conjugate operator $X_0$. However, the most natural choice is to take the \textbf{Liouvillian} $L$ defined by 
$$
L\Omega=0,
$$
where $\Omega\in\Hi_\omega$ is the vector such that $\omega(A)=\bra\Omega|A\Omega\ket$ for all $A\in\Mi$. If $\Mi=\Bi(\Hi_0)$ for some Hilbert space $\Hi_0$ then 
\begin{equation}\label{Liouvtilde}
L=H\otimes\bone-\bone\otimes H=H-\tilde{H}
\end{equation}
acting on $\Hi_\omega=\Hi\otimes\overline{\Hi}=\Li^2(\Hi_0)$, where $\Li^2(\Hi_0)$ is the algebra of so-called Hilbert-Schmidt operators, $H$ is the Hamiltonian and $\tilde{H}$ is equal to $H$ acting from the right. We shall not need this explicit realization of $\Hi_\omega$ but \eqref{Liouvtilde} shows that the spectrum of $L$ is the set of energy differences (the transition frequencies of the system). The Liouvillian plays an important role in most applications of open quantum systems \cite{BP}, e.g. in nonlinear spectroscopy \cite{Muk}.

In many interesting cases $L$ has absolutely continuous spectrum $\R$ plus the single isolated eigenvalue $0$ corresponding to $\Omega$ embedded in this absolutely continuous spectrum \cite{tBW}. Consider the Hilbert space decomposition
\begin{equation}\label{Paulieq}
\Hi_ \omega=\C\Omega\oplus(\C\Omega)^\perp.
\end{equation}
Restricting $L$ off the one-dimensional subspace $\C\Omega$ spanned by $\Omega$ it is unitarily equivalent to multiplication by $x$ on $L^2(\R,d\nu)$, where $d\nu(x)$ is a measure which is absolutely continuous with respect to the Lebesgue measure. Hence it is possible to find an operator $X_0$ on $(\C\Omega)^\perp$ such that
$$
[L,X_0]=i\bone.
$$ 
This will be the generator of energy translations. With respect to the decomposition \eqref{Paulieq} we can write any operator $A$ on $\Hi_\omega$ as 
$$
A=\begin{pmatrix}A_{11}&A_{12} \\A_{21}&A_{22}\end{pmatrix}.
$$
The deformations of $L$ which we obtain using $X_0$ as generator will thus be of the form $L+V_{22}$ with $V_{22}$ an operator on $(\C\Omega)^\perp$. 

In short, as long as we stay in the equilibrium representation defined by an equilibrium state on the observable algebra, the Pauli theorem is automatically circumvented. This will be used in the next subsection.

\begin{Remark}
If $X_1,X_2,X_3$ and $P_1,P_2,P_3$ are second quantizations acting on Fock space then we have $[X^j,P_k]/i=\delta^j_k N$ where $N$ is the number operator. Hence there is a subspace $\C\Omega$ on which $[X^j,P_k]=0$ and we are in the same situation as above. Both $X_j$ and $P_j$ must annihilate the vacuum $\Omega$.  
\end{Remark}

\subsection{Deforming the Liouvillian}
Let us from now on study interactions with quantum systems which are in an equilibrium state. 

Suppose observables of a given quantum system are represented by operators acting in the equilibrium Hilbert space $\Hi_\omega$. The space-time translations on the system are mathematically described as a unitary group on $\Hi_\omega$ parameterized by $\R^4$. Motivated by the discussion in §\ref{Paulisec}, we assume that the generators $(P_\mu)=(L,P_1,P_2,P_3)$ of the unitary space-time  transformations have purely absolutely continuous spectrum except for an isolated zero corresponding the the vector $\Omega$ implementing the state $\omega$. We then define $(X^\mu)=(X^0,X^1,X^2,X^3)$ to be operators on $(\C\Omega)^\perp$ which satisfy $[P_\mu,X^\nu]/i=\delta_\mu^\nu\bone$ . 

Let $\Ki$ be another Hilbert space and consider the unitary on $\Hi_\omega\otimes\Ki$ given by
\begin{align*}
W_{e}&=e^{-ie\Theta_{\mu\nu}X^\mu\otimes X^\nu},\quad\quad 
\Theta_{\mu\nu}:=\begin{pmatrix}0&-E_1&-E_2&-E_3\\
E_1&0&0&0\\
E_2&0&0&0\\
E_3&0&0&0\\
\end{pmatrix}
\end{align*}
for some real constants $E_1,E_2,E_3$ and $e$, where the $X^\nu$'s acting in $\Ki$ are assumed to have properties analogous to the $X^\mu$'s acting in $\Hi_\omega$ (thus the space $\Ki$ is also describing something like a system in equilibrium).
At this point it is useful to recall how we reasoned before: a measurement of the coordinate operators of a quantum system corresponds to boosting up the system (nonrelativistically). 
\begin{prop}\label{potprop}
For any normal state $\varphi$ on $\Bi(\Ki)$ with $\varphi(X^\mu)<+\infty$ for all $\mu=1,2,3,4$, the measurement $(\Ki,\varphi,W_e)$ of the coordinate operators $X^\mu$ on $\Hi_\omega$ deforms the Liouvillian $L=L_\omega$ into 
\begin{align*}
L_{e\Theta}&:=(\iota\otimes\varphi)[W_{e}^{-1}(L\otimes\bone)W_{e}]=L+e\int_\R \varphi[\textnormal{d}E^{\bone\otimes\mathbf{X}}(\mathbf{x})]\mathbf{E}\cdot\mathbf{x},
\end{align*}
where $\mathbf{E}:=(E_1,E_2,E_3)$ and $\mathbf{E}\cdot\textbf{x}:=E_1x_1+E_2x_2+E_3x_3$. Applying $\Ad(W)$ to the momenta gives
\begin{align*}
e^{i\Theta_{\mu\nu}X^\mu\otimes X^\nu}(P_k\otimes\bone)e^{-i\Theta_{\mu\nu}X^\mu\otimes X^\nu}&=e^{-iE_kX^k\otimes X^0}(P_k\otimes\bone)e^{iE_kX^k\otimes X^0}
\\&=P_k\otimes\bone-E_k\bone\otimes X^0.
\end{align*}
\end{prop} 
Defining the potential energy as $V=e\Phi:=\textbf{E}\cdot\varphi(\textbf{X})$, the energy operator in the presence of a constant electric field would for subsequent measurement on $\Hi_\omega$ appear as
$$
L_{e\Theta}=L+V,
$$
and $\Hi_\omega$ cannot be separated from $\Hi_\omega\otimes\Ki$ under such conditions (the electric field system $\Ki$ and the system $\Hi_\omega$ appear ``entangled" to an observer interacting via slower energy transfer). 

The classical electromagnetic field (using same notation), 
$$
\textbf{E}=-\boldsymbol\nabla\Phi-\frac{\partial \textbf{A}}{\partial t},
$$
thus manifests itself via the vector potential $\textbf{E}\bone\otimes X^0$ and the scalar potential $\Phi=V/e$. Still, the operator $X^0$ is probably best regarded just as a generator of energy transfer. 

Adding Proposition \ref{potprop} to the discussion in §\ref{Minimal Coupling to Gauge Potential}, we conclude that the interaction of a quantum system with an electromagnetic field is recovered using an evolution operator $e^{-iF_{\mu\nu}X^\mu\otimes X^\nu}$ where the skew-symmetric matrix $F_{\mu\nu}$ is none other than the electromagnetic force:
\begin{equation}\label{force}
F_{\mu\nu}:=\begin{pmatrix}
0&-E_1&-E_2&-E_3\\
E_1&0&-B_3&B_2\\
E_2&B_3&0&-B_1\\
E_3&-B_2&B_1&0\\
\end{pmatrix}.
\end{equation}

\subsection{Gauge structure}\label{gaugesec}
If $\theta$ is a differentiable real-valued function on $\R^4$ then the unitary $U:=e^{i\theta(X)}$ satisfies
\begin{equation}\label{gaugetransfo}
U\mathbf{P}U^{-1}=\mathbf{P}-\boldsymbol\nabla \theta,\qquad ULU^{-1}=L-\frac{\pd \theta}{\pd t},
\end{equation}
and the system is said to be \textbf{gauge invariant} under the transformations \eqref{gaugetransfo} precisely because they are implemented by a unitary in the system Hilbert space $\Hi_\omega$. The motivation for this is if we apply $U$ to vectors $\psi\in\Hi_\omega$ at the same time as \eqref{gaugetransfo}, everything remains unchanged \cite[§5]{Ja}.

The same transformations \eqref{gaugetransfo} can be achived by the replacement
\begin{equation}\label{gaugedgener}
A_\mu\to A_\mu+\pd_\mu \theta(X),\qquad \mu=0,1,2,3
\end{equation}
in the interaction $W=e^{-iX^\mu\otimes A_\mu}\in\Bi(\Hi_\omega\otimes\Ki)$. We see that replacements such as \eqref{gaugedgener} are precisely those which can be implemented by ``local unitary transformations", i.e. by unitaries acting on $\Hi_\omega$ alone. They can be achived without an external interaction and should be regarded as a gauge degree of freedom in the system.

\section{Classical limit}\label{quantization}

\subsection{Phase-space quantum mechanics and Rieffel deformations}
Quantization requires a choice of ordering of operators. Let us recall some aspects of ``Weyl quantization", which corresponds to the symmetric ordering. For a nice phase-space function $f:\R^n\times\R^n\to\C$, the Weyl quantization $\Op(f)$ is an operator on $L^2(\R^n)$. For $f,g\in\Si(\R^n\times\R^n)$ (Schwartz space), the product $\Op(f)\Op(g)$ is again a Weyl operator. Therefore, composition of Weyl operators defines implicitly a noncommutative product $\times_\Theta$ on an algebra of functions on $\R^n\times\R^n$:
$$
\Op(f)\Op(g)=\Op(f\star g).
$$
This product (called the ``Moyal product") has an explicit integral formula
$$
(f\star g)(x)=\iint_{\R^{2n}\times \R^{2n} }f(x+\Theta z)g(x+y)e^{2\pi i z\cdot y}\, dz\, dy
$$
where $\Theta=\bigl(\begin{smallmatrix}0&\bone\\ -\bone&0\end{smallmatrix} \bigr)$ is the standard symplectic structure on $\R^n$. If one defines an action $\alpha$ of $\R^n\times\R^n$ on $\Si(\R^n\times\R^n)$ by translations,
$$
(\alpha_z(f))(x):=f(x+z),
$$
then $f\star g$ can be expressed as
$$
f\star g=\iint{\R^{2n}\times \R^{2n}}\alpha_{\Theta z}(f)\alpha_y(g)e^{2\pi i z\cdot y}\, dz\, dy.
$$
Rieffel deformation \cite{Rie} amounts to defining such a deformed product on a general $C^*$-algebra $\Ai$. He shows that it has nice properties and makes sense also if $\Ai$ is not commutative. 

In fact, if $\alpha$ is an action of $\R^n$ on an algebra of observables $\Ai$ (a $C^*$-algebra would suffice), the deformation of operators that we obtained from measurements (``warped convolution") can be effected by keeping the \emph{undeformed} operators and instead introducing a \emph{deformed product} $\times_\Theta$, where $\Theta$ is the same deformation matrix as the one used for warped convolution. The elements $A,B$ smooth under the action $\alpha$ on $\Ai$ form an algebra under $\times_\Theta$, which is explicitly given by
\begin{equation}\label{defproduct}
A\times_{\Theta}B=\int_{\R^n \times\R^n}\alpha_{\Theta x}(A)\alpha_y(B)e^{2\pi i x\cdot y}\, dx\, dy.
\end{equation}
If $A_\Theta$ denotes the warped convolution of an operator $A$ with respect to the same action $\alpha$ and matrix $\Theta$, then the important relation is (see \cite{BLS})
\begin{equation}\label{prodrelation}
A_\Theta B_\Theta=(A\times_{\Theta}B)_\Theta.
\end{equation}
For more information about the relation between Rieffel deformation and warped convolution, see \cite{An2}. 

To each such product $\times_\Theta$ there is an associated Poisson bracket on the subalgebra of $\alpha$-smooth elements.  
\begin{Definition}\label{Poisson bracket}
Let $\alpha$ be an action of $\R^n$ on a $C^*$-algebra $\Ai$ and let $\Theta$ be a skew-symmetric $n\times n$ matrix. We say that a Poisson bracket $\{,\cdot,\}_\Theta$ has the same direction as the deformation \eqref{defproduct} if it corresponds to the first-order bidifferential operator in the asymptotic expansion of the product \eqref{defproduct}.
\end{Definition}

Explicitly, for elements $A,B\in\Ai$ smooth under $\alpha$, such a bracket is given by \cite[\S6]{Rie}
\begin{equation}\label{assbracket}
\{A,B\}_\Theta=\sum_{j,k=1}^n\Theta_{j,k}(X_jA)(X_kA),
\end{equation}
where $X_j$ denotes differentiation in the $j$'th direction of $\R^n$ via $\alpha$. 

When $\Theta$ is the standard symplectic structure on $\R^n$ the Poisson bracket is just the canonical one. As we discuss in §\ref{Feynman Bracket}, the case when $\Theta$ is the force matrix \eqref{force} gives a Poisson bracket which resembles the one in Feynman's approach to electromagnetism.

\subsection{Feynman brackets meet Rieffel deformations}\label{Feynman Bracket}
Let us now show that Maxwell's theory of electromagnetism arises as the classical limit of the interactions between two systems we have modeled according to quantum measurement theory. 

First we recapitulate why a deformation of the operator multiplication could be a good approximation for the effect of a measurement interaction. We always discuss the interaction 
$$
F_{\mu\nu}X^\mu\otimes (X^\nu+\pd_\mu\theta(X))=X^\mu\otimes A_\mu
$$
where $A=(A_0,A_1,A_2,A_3)$ is the gauge potential and $\theta$ is as in §\ref{gaugesec}. 
\begin{enumerate}[(i)]
\item{Measurement using $e^{-iX^\mu\otimes A_\mu}$ deforms the operators on $\Hi$ via a very peculiar warped convolution, namely one which is unitary implemented. }
\item{There is always a version, corresponding to ignoring the second system $\Ki$, which is obtained from another warped convolution. Namely, it comes from the action $\alpha_p(T):=e^{ip_\mu X^\mu}Te^{-ip_\mu X^\mu}$ on $T\in\Bi(\Hi)$ and defines $T_\Theta$ by \eqref{warped} where $\Theta=F$ refers to the constant force matrix \eqref{force}}.
\item{The relation $S_\Theta T_\Theta=(S\times_\Theta T)_\Theta$ defines a deformed product $\times_\Theta$ on the algebra of smooth elements in $\Bi(\Hi)$.}
\end{enumerate}
We can describe a spacetime variation in $A_\mu$ if we smear the field with some smooth function $g\in\Si(\R^4)$ (``form factor" in $x$-space), defining
$$
A_\mu(g):=\int_{\R^4}g(x)_\mu F_{\mu\nu}(x^\nu+\pd_\mu\theta(x))\, dx,
$$
where summation is understood for vectors $(x_\mu)$ in expressions like $x_\mu x^\mu $ but \emph{not} in $x_\mu x_\mu$. 

For the purpose of discussing the classical limit, we now consider observables parameterized by space-time, in the spirit of quantum field theory. Let $\Lambda$ be a region in $\R^4$ containing the system of interest. For each subregion $\Oi\subseteq\Lambda$ we associate a $C^*$-algebra $\Ai(\Oi)$ containing the observables whose expectation values are, in the absence of interactions, functions of points in $\Oi$ only (i.e. they do not depend on the points outside $\Oi$). We thus have $\Ai(\Oi_1)\subseteq\Ai(\Oi_2)$ whenever $\Oi_1\subseteq\Oi_2$. Each $\Ai(\Oi)$ is supposed to act on the same Hilbert space $\Hi$. Now we want to describe an interaction with the system $\Oi\subset\Lambda$ of the form $W(\Oi)=e^{-iX^\mu\otimes A_\mu(g)}$, where the tensor product $\Hi\otimes\Ki$ contains the state vectors of the total system and the form factor $g$ is supported in $\Oi$.  Let
$$
\alpha_p(X):=e^{ip_\mu X^\mu}Xe^{-ip_\mu X^\mu},\qquad X\in\Bi(\Hi),\,p\in\R^4
$$
be the associated action on $\Bi(\Hi)$ that we get by ignoring the second system $\Ki$. 

We may let $\Oi$ be sufficiently small for the interaction unitary $W$ to not vary over the size of $\Oi$ (the ``dipole approximation" with respect to the size of $\Oi$). Then for any point $x\in\Oi$, the interaction with $\Ai(\Oi)$ is given by $W(x)=e^{-iX^\mu\otimes A_\mu(x)}=e^{-iF_{\mu\nu}(x)X^\mu\otimes X^\nu}$, where $F(x)$ is the force matrix \eqref{force} with entries $F_{\mu\nu}(x)$ determined by this fixed $x$. For elements $S$, $T$ of $\Ai(\Oi)$ which are smooth for the action, the deformed product is
$$
S\times_{F(x)}T=\iint_{\R^4 \times\R^4}\alpha_{F(x)y}(S)\alpha_{q}(T)e^{2\pi iq\cdot y}\,dy\, dq.
$$
We require that $\Ai(\Oi)$ contains all operators $f(P)$ for $f\in C_0(\R^4)$ which are functions of the momenta $P=(P_\mu)$ and whose Fourier transform is supported in $\Oi$. For two such operators $f_1(P)$ and $f_2(P)$ we have
 $$
f_1(P)\times_{F(x)}f_2(P)=\iint_{\R^4 \times\R^4}f_1(P+F(x)y)f_2(P+q\bone)e^{2\pi iq\cdot y}\,dy\, dq.
$$
To obtain a Poisson bracket on $\Ai(\Oi)$ we first consider the commutative subalgebra generated by the $f(P)$'s. The momenta $P_0,P_1,P_2,P_3$ correspond to the coordinate functions $p_0,p_1,p_2,p_3$ on momentum space $\R^4$. Recall the standard definition of a Poisson bracket on an algebra of smooth functions on a Lie group \cite[\S3.1]{dSW},
$$
\{f,g\}(x):=\sum_{\mu,\nu=1}^4\{p_\mu,p_\nu\}(x)\frac{\partial f}{\partial p_\mu}\frac{\partial g}{\partial p_\nu},\quad\quad f,g\in C^\infty(\R^4).
$$
If this is supposed to be  in the same direction as the deformation (see Definition \ref{assbracket}) then comparing with the right-hand side of Equation \eqref{assbracket} we see that
\begin{equation}\label{Fbracket}
\{p_\mu,p_\nu\}(x)=F_{\mu\nu}(x).
\end{equation} 
The coordinate operators $X^0,X^1,X^2,X^3$ are defined to be operators on $\Hi$ satisfying $[X^\mu,P_\nu]=i\delta^\mu_\nu\bone$. If the Fourier transform of $f\in C_0(\R^4)$ is supported in $\Oi$ then $[h(X),f(P)]$ is supported on $\operatorname{supp}(h)\cap\Oi$. In that way, the parameterization $x=(x^\mu)$ of space-time $\R^4$ is given by the spectral values of $X=(X^\mu)$. 

For operators on $\Bi(\Hi)$ which are sufficiently smooth under the actions of both $e^{-ix^\mu P_\mu}$ and $e^{-p_\mu X^\mu}$ we can replace the operator multiplication by the Moyal product. These operators are thus identified with functions on $\R^4\times\R^4$. Taking the limit $\hbar\to 0$ we get a bracket including the term \eqref{Fbracket} as well as the standard term $\{x^\mu,p_\nu\}=\delta^\mu_\nu 1$. This is the classical limit of $\Ai$ obtained when both $F$ and $\hbar$ go to zero. On the classical level, the brackets \eqref{Fbracket} define a function of $x\in\R^4$. 

There has been many papers discussing the derivation of Maxwell's equations using nonstandard Poisson brackets (the most inspiring for this work being \cite{MP},\cite{Bra}). According to Dyson \cite{Dys}, the first ideas of this sort are Feynman's, hence the name ``Feynman brackets", even though Feynman assumed a commutator instead of a Poisson brackets, and his nonrelativistic setting made it hard to identify the force components $F_{\mu\nu}$ as the structure functions of such a bracket. We could use their arguments to take the final steps to the Maxwell equations. However, the way we found $F$ makes this unnecessary.

For, we started from the unitary $e^{-iF_{\mu\nu}X^\mu\otimes X^\nu}$ and defined the magnetic and electric fields $\textbf{B},\textbf{E}$ as the components of $F$, identifying at the same time the relations $\textbf{B}=\nabla\times\textbf{A}$ and $\textbf{E}=-\nabla\Phi-\partial_t\textbf{A}$. As is well known, this is enough to obtain the field equations. Explicitly, with $x^0,x^1,x^2,x^3$ the dual basis of $\R^4$ for which the coordinate operators (and hence $\nabla$) corresponds, the force satisfies
$$
\frac{\partial F_{\nu\rho}}{\partial x^\mu}+\frac{\partial F_{\rho\mu}}{\partial x^\nu}+\frac{\partial F_{\mu\nu}}{\partial x^\rho}=0,
$$
which gives the two homogeneous field equations. Next, $\partial^\mu \partial^\nu F_{\mu\nu}=0$ implies
$$
\partial^\nu F_{\mu\nu}=j_{\mu}
$$
for some source field $j_{\mu}$. This gives the remaining two equations.

The aim of this section was not to derive classical electromagnetism using a minimal set of mathematical assumptions. It is hard to beat Feynman and followers in this aspect. Rather, the point is that the interaction $e^{-iF_{\mu\nu}X^\mu\otimes X^\nu}$ represents quantum disturbances which accumulate to manifest themselves as classical forces. There was a \emph{physical} motivation for introducing the force based on the most basic model for interactions (viz. the von Neumann model). 

\section{Concluding remarks}
We have seen the relevance of the notion of sequential measurement in a simple example where it accounts for the effect of a magnetic field on the quantum system of a charged particle moving in space. There the interaction with the magnetic field constitutes a first measurement and subsequent measurements done by some experimentalist will be affected by that interaction. It is the much faster timescale of the magnetic-field interaction with the particle that makes us view it as the interaction ``prior" to the subsequent one performed by some experimentalist. The result is that the experimentalist can only get a correct description if he includes a minimal coupling term in the Hamiltonian. Such a simultaneous measurement also requires that the measured energy observable is a coarse-grained version of the Hamiltonian. Our main insight is thus that such a discretization appears from the competition of the interactions described by the same theory. 

We have assumed that the energy-momentum $P=(L,P_1,P_2,P_3)$ has purely absolutely continuous spectrum except for an isolated zero eigenvalue for a common eigenvector. Such a setting appears for free particles or in representations of equilibrium states on local observable algebras in algebraic quantum field theory. In particular we used an infinite-dimensional Hilbert space $\Hi$. Although usually not needed in practice, we have thus added infinitely many degrees of freedom to make it possible for electromagnetism to appear using only Heisenberg commutation relations.

To make contact with ordinary quantum mechanics where the system is usually taken to be a finite-dimensional Hilbert space $\GH\cong\C^n$, we mention that such a setting amounts to considering $\Hi\otimes\GH$, which can also be identified with $\Hi^{\oplus n}$, the direct sum of $n$ copies of $\Hi$. 

The explanation of quantum discreteness that we gave in §\ref{discsec} is just the quantum Zeno effect \cite{FP1}, \cite{FP2}. Namely, the latter can also be understood as resulting from two competing interactions of a system interacting with two different quantum systems. What we have done is simply to isolate what such competing interaction is needed in order to get electromagnetism as classical limit of a one-level system. 

\section{References}
\bibitem[ALV]{ALV} Accardi L, Lu YG, Volovich I. Quantum theory and its stochastic limit.  Springer-Verlag (2002).

\bibitem[An1]{An1} Andersson A. Operator deformations in quantum measurement theory. Lett. Math. Phys. Vol 104, Issue 4, pp. 415-430 (2014). 

\bibitem[An2]{An2} Andersson A. Index pairings for $\R^n$-actions and Rieffel deformations. arXiv:1406.4078 (2014). 

\bibitem[ADK]{ADK} Aschieri P, Dimitrijevi\'c M, Kulish P, Lizzi F, Wess J. Noncommutative spacetimes: symmetries in noncommutative
geometry and field theory. Lect. Notes Phys. Vol 774 (2009).

\bibitem[BaGh]{BaGh} Basu B, Ghosh S. Quantum Hall effect in bilayer systems and the noncommutative plane: A toy model approach. Physics Letters A. Vol 346, Issue 1, pp. 133-140 (2005).

\bibitem[Bra]{Bra} Bracken P. Relativistic equations of motion from Poisson brackets. International Journal of Theoretical Physics. Vol 37, Issue 5, (1998). 

\bibitem[BP]{BP} Breuer HP, Petruccione F. The theory of open quantum systems. Oxford University Press (2003).

\bibitem[BLS]{BLS} Buchholz D, Lechner G, Summers S. Warped convolutions, Rieffel deformations and the construction of quantum field theories. Commun. Math. Phys. Vol 304, pp. 95-123 (2011). 

\bibitem[BS]{BS} Buchholz D, Summers S. Warped convolutions: a novel tool in the construction of quantum field theories. In: Quantum field theory and beyond. World Scientific, pp. 107-121 (2008). 

\bibitem[BL]{BL} Busch P, Lahti PJ. The standard model of quantum measurement theory: history and applications. Found. Phys. Vol 26, Issue 7, pp. 875-893 (1996).

\bibitem[BLM]{BLM} Busch P, Lahti PJ,  Mittelstaedt P. The quantum theory of measurement. Second revised edition. Springer-Verlag (1996).

\bibitem[CHT]{CHT} Carmeli C, Heinonen T, Toigo A. On the coexistence of position and momentum observables. J. Phys. A. Vol 38, pp. 5253-5266 (2005).

\bibitem[CHT2]{CHT2} Carmeli C, Heinonen T, Toigo A. Sequential measurements of conjugate observables. J. Phys. A: Math. Theor. Vol 44, p. 285304 (2011).

\bibitem[dSW]{dSW} Da Silva AC. Weinstein A. Geometric models for noncommutative algebras. Berkeley Mathematics Lecture Notes Vol 10 (1999).

\bibitem[DaJe]{DaJe} Dayi ÖF, Jellal A. Hall effect in noncommutative coordinates. J. Math. Phys. Vol 43, Issue 10, pp. 4592-4601 (2002).

\bibitem[DL]{DL} Davies EB, Lewis JT. An operational approach to quantum probability. Comm. Math. Phys. Vol 18, p. 239 (1970).

\bibitem[Dys]{Dys} Dyson FJ. Feynman's proof of the Maxwell equations. Amer. J. Phys. Vol 58, p. 209 (1990).

\bibitem[FP1]{FP1} Facchi P, Pascazio S. Quantum Zeno dynamics: mathematical and physical aspects. J. Phys. A: Math. Theor. Vol 41, Issue 49, p. 493001. (2008).

\bibitem[FP2]{FP2} Facchi P, Pascazio S. Three different manifestations of the quantum Zeno effect. In: Irreversible Quantum Dynamics. Lect. Notes Phys. Vol 622, p. 141. Springer-Verlag (2003).

\bibitem[Haa]{Haa} Haag R. Local quantum physics. Springer-Verlag (1992).

\bibitem[HHW]{HHW} Haag R, Hugenholtz NM, Winnink M. On the equilibrium states in quantum statistical mechanics. Commun. Math. Phys. Vol 5, pp. 215-236 (1967).

\bibitem[HLY]{HLY} Heinonen T, Lahti P, Ylinen K. Covariant fuzzy observables and coarse-graining. Rep. Math. Phys. Vol 53, Issue 3, pp. 425-441 (2004).

\bibitem[HRS]{HRS} Heinosaari T, Reitzner D, Stano P. Notes on joint measurability of joint observables. Found. Phys. Vol 38, Issue 12, pp. 1133-1147 (2008).

\bibitem[HW]{HW} Heinosaari T, Wolf M. Nondisturbing quantum measurements. J. Math. Phys. Vol 51, Issue 9, p. 092201 (2010).

\bibitem[Hol]{Hol} Holevo AS. Radon-Nikodym derivatives of quantum instruments. J. Math. Phys. Vol 39, pp. 1373-1387 (1998).

\bibitem[Ja]{Ja} Jauch JM. Gauge invariance as a consequence of Galilei-invariance for elementary particles. Hel. Phys. Acta. Vol 37, pp. 284-292 (1964).

\bibitem[Lan]{Lan} Landsman NP. Mathematical topics between classical and quantum mechanics. Springer-Verlag (1998). 

\bibitem[LW]{LW} Lechner G, Waldmann S. Strict deformation quantization of locally convex algebras and modules. arXiv:1109.5950 (2011).

\bibitem[MP]{MP} Montesinos M, P\'{e}rez-Lorenzana A. Minimal coupling and Feynman's proof. Int. J. Theor. Phys. Vol 38, Issue 3, pp. 901-910 (1999).

\bibitem[Mor]{Mor} Moretti V. Spectral theory and quantum mechanics - with an introduction to the algebraic formulation. Unitext, La Matematica per il 3+2. Springer (2013). 

\bibitem[Mu1]{Mu1} Much A. Quantum mechanical effects from deformation theory. Quantum mechanical effects from deformation theory. J. Math. Phys. Vol 55, Issue 2, p. 022302 (2014).

\bibitem[Mu2]{Mu2} Much A. Wedge-local quantum fields on a nonconstant noncommutative spacetime. J. Math. Phys. Vol 53, Issue 8, p. 082303 (2012).

\bibitem[Mu3]{Mu3} Much A. Self-adjointness of deformed unbounded operators. J. Math. Phys. Vol 56, p. 093501 (2015).

\bibitem[MME]{MME} Muga J, Sala Mayato R, Egusquiza I (eds.). Time in quantum mechanics. Lect. Notes Phys. Vol 72. Springer (2002).

\bibitem[Muk]{Muk} Mukamel S. Principles of nonlinear optical spectroscopy. Oxford University Press (1995).

\bibitem[Oza]{Oza} Ozawa M.  Quantum measuring processes of continuous observables. J. Math. Phys. Vol 25, Issue 1, pp. 79-87 (1984).

\bibitem[Pau]{Pau}  Pauli W. Handbuch der Physik, Vol 24, part 1, p. 140, Springer-Verlag (1958).

\bibitem[Rie]{Rie} Rieffel MA. Deformation quantization for actions of $\R^d$. Mem. Amer. Math. Soc. Vol 106, Issue 506 (1993).

\bibitem[tBW]{tBW} Ten Brinke G, Winnink M. Spectra of Liouville operators. Commun. Math. Phys. Vol 51, pp. 135-150 (1976).

\bibitem[vN]{vN} Von Neumann J. Mathematical foundations of quantum mechanics. Princeton University Press (1955).

\end{document}